\documentclass[1p]{elsarticle}

\usepackage{amssymb}
\usepackage{amsmath}
\usepackage{graphicx}
\usepackage{amsthm}
\usepackage{color}
\usepackage{mathrsfs}
\usepackage{tikz}
\usetikzlibrary{arrows,snakes,backgrounds,positioning}

\definecolor{grey}{rgb}{0.75,0.75,0.75}
\newtheorem{defi}{Definition}
\newtheorem{prop}{Proposition}

\newtheorem{theorem}{Theorem}

\DeclareMathOperator{\ID}{id}
\DeclareMathOperator{\Des}{\mathcal{D}}
\DeclareMathOperator{\Aug}{{\rm Aug}}

\author[ligm]{Gregory Kucherov}
\ead{Gregory.Kucherov@univ-mlv.fr}

\author[ligm,ELTE]{Lilla T\'othm\'er\'esz\corref{cor}}
\ead{tothmereszlilla@gmail.com}

\author[ligm]{St\'ephane Vialette}
\ead{vialette@univ-mlv.fr}

\cortext[cor]{Corresponding author}
\address[ligm]{Universit\'e Paris-Est \& CNRS, Laboratoire d'Informatique Gaspard Monge, Marne-la-Vall\'ee, France}
\address[ELTE]{Lor\'and E\"otv\"os University, P\'azm\'any P\'eter s\'et\'any 1/C, H-1117 Budapest, Hungary}

\title{On the combinatorics of suffix arrays}

\begin{document}

\begin{abstract}
We prove several combinatorial properties of suffix arrays, including
a characterization of suffix arrays through a bijection with a certain
well-defined class of permutations. Our approach is based on the
characterization of Burrows-Wheeler arrays given in \cite{Cro05}, that
we apply by reducing suffix sorting to cyclic shift sorting through
the use of an additional sentinel symbol. 
We show that the
characterization of suffix arrays for a special case of binary
alphabet given in \cite{He2005} easily follows from our
characterization. Based on our results, we also provide simple
proofs for the enumeration results for suffix arrays, obtained in
\cite{Schurmann2008}. Our approach to characterizing suffix arrays is
the first that exploits their
relationship
with Burrows-Wheeler permutations.


\end{abstract}
\maketitle

\section{Introduction}

Suffix array is a very popular data structure in string algorithms,
both in theoretical studies and practical applications,
that has been designed as a space-efficient alternative to suffix
trees
\cite{DBLP:journals/siamcomp/ManberM93,DBLP:books/ph/frakesB92/GonnetBS92}.
With the discovery of linear-time construction algorithms
for suffix arrays
\cite{DBLP:conf/icalp/KarkkainenS03,DBLP:conf/cpm/KimSPP03,DBLP:conf/cpm/KoA03},
this data structure received an increasing attention during the last decade. A good deal of work
has been devoted to improving the practical efficiency of suffix
arrays.

A suffix array for a string of length $n$ is essentially a
permutation of length $n$ corresponding to the starting positions of
all suffixes sorted lexicographically. Obviously, if the alphabet has
a fixed size $k<n$,
then only a proper subset (at most $k^n$) of all $n!$ permutations are suffix arrays
for some word over this alphabet.
A main motivation of this paper is to provide a characterization for
suffix array permutations for bounded-size alphabets.

Our results take advantage of a very close relation between suffix arrays
and the {\em Burrows-Wheeler transform} \cite{BW94}. The
Burrows-Wheeler transform of a string is a permutation of string
letters which allows the string to be effectively reconstructed. Among other applications, it is the basis for many {\em
  compact text indexes} that have been intensively studied and used in
practical applications (see e.g. \cite{DBLP:journals/csur/NavarroM07}
and references therein).

Crochemore et al. \cite{Cro05} pointed out a
very nice characterization of {\em Burrows-Wheeler arrays}, which are
close relatives of suffix arrays. This characterization, attributed to
Gessel and Reutenauer \cite{DBLP:journals/jct/GesselR93}, uses the
key notion of {\em linking permutation} which is similar to (but
different from) the notion of $\Psi$-function studied for suffix
arrays \cite{Grossi}. We show that the approach of \cite{Cro05} can be
successfully applied to obtain characterization results for suffix
arrays, using a relation between orderings of suffixes and cyclic shifts.

As far as related works are concerned, He at al. \cite{He2005}
provided a characterization of suffix arrays
for the case of binary alphabet ($k=2$) and an assumption
that the terminal sentinel symbol is ranked between the two main
symbols in the alphabet ordering.
We show that the characterization of
\cite{He2005} easily follows from the characterization that we
propose in this paper.

In \cite{Schurmann2008}, Sch\"urmann and Stoye prove several counting
results for suffix arrays and corresponding strings. They
use a characterization of suffix array permutations through
$\Psi$-functions, that they call $R_+$-arrays, which are mappings from
$[1,n]$ to $[0,n]$. A crucial parameter in countings is the number of
{\em descents} in $R_+$-arrays, directly related to the minimal
alphabet size on which the corresponding suffix array can be realized (see also \cite{DBLP:conf/mfcs/BannaiIST03}).
Compared to our approach,
the important difference is that the set of $R_+$-arrays is not
characterized, while the set of linking permutations admit a
neat combinatorial characterization as permutations with only one
orbit. This allows us to provide a bijection between suffix arrays and
a certain well-defined class of permutations. To demonstrate the
usefulness of our approach, we obtain
much simpler proofs of counting theorems from \cite{Schurmann2008}.
Our approach to characterizing suffix arrays is
the first that exploits their
relationship
with Burrows-Wheeler permutations.

\section{Preliminaries}
\label{prelim-sect}

In what follows, $\Sigma=\{a_1,a_2,\ldots,a_k\}$ is an ordered
alphabet of size $k$, where $a_1<a_2<\dots <a_k$.
The set of all permutations
$\pi= \pi(1)\, \pi(2) \, \dots\, \pi(n)$ of length $n$ is denoted by
$\mathbf{S}_n$, and $\ID$ denotes the identity permutation
$1\, 2\, \ldots\, n$.
The composition of permutations $\sigma$ and $\pi$ is
denoted $\pi\sigma$, \emph{i.e.},
$(\pi\sigma)(i)=\pi(\sigma(i))$.
Throughout the paper, we assume that the addition (subtraction) of a constant value
to a permutation value verifies the identities $n+1\equiv 1$,
$1-1\equiv n$.
In other words, $\pi(i)+k = ((\pi(i)-1+k)\mod n)+1$.
For a permutation $\pi\in\mathbf{S}_n$, we let $\pi+k$, $k \in [1,n]$, stand for the
permutation defined by
$(\pi+k)(i)=\pi(i)+k$.
Note that $\pi+k=(\ID+k)\pi$. The permutation $(\pi-k)$ is defined similarly.

\begin{defi}[suffix array]
 Given a word $w=w_1w_2\dots w_n$
on alphabet $\Sigma$,
 its {\em suffix array} is a permutation $\pi$
such that
$\pi(i)=j$ iff the suffix $w_j
\dots w_n$ is the $i^{th}$ in the lexicographic ordering of all suffixes
 of $w$.
\end{defi}

For example, the suffix array of $babba$ is $5\,2\,4\,1\,3$.

\begin{defi}[primitive word]
A word $u \in \Sigma^+$ is called \emph{primitive} if it is not a proper
power of another word,
\emph{i.e.},
$u = v^n$, $v \in \Sigma^+$, implies $n = 1$.
\end{defi}

Primitive words are exactly those words whose cyclic shifts are all
distinct. Therefore, for a primitive word, we can consider the
permutation defined by the lexicographic ordering of its cyclic
shifts. We call this permutation the \emph{Burrows-Wheeler array}
because of its direct relation to the Burrows-Wheeler
transform~\cite{BW94}.

\begin{defi}[BW-array]
Given a primitive word $w=w_1\dots w_n$,
its {\em Burrows-Wheeler array} (hereafter {\em BW-array}) is a
permutation $\pi$
such that
$\pi(i)=j$ iff the word $w_j \dots w_nw_1\dots w_{j-1}$
is the $i^{th}$ in the lexicographic ordering of all cyclic shifts of
$w$.
\end{defi}

For example, the BW-array of $bbaba$ is $3\, 5\, 2\, 4\, 1$.

We write $\mathbf{S}_n^c$ for the set of all permutations of
$\mathbf{S}_n$ with one orbit.
The following notion has proved to be very helpful for
characterizing BW-arrays \cite{DBLP:journals/jct/GesselR93,Cro05}. It
is related to $\Psi$-functions  \cite{Grossi} or $R_+$-arrays
\cite{Schurmann2008} defined on suffix arrays, but defines a mapping
on permutations.

\begin{defi}[linking permutation, linking mapping]
Let $\pi \in \mathbf{S}_n$.
The \emph{linking permutation} of $\pi$ is
the permutation $\varphi=\pi^{-1}(\pi+1) \in \mathbf{S}_n^c$.
The mapping $\pi \mapsto \varphi$ is called the \emph{linking mapping},
and is denoted by $\Phi$.
\end{defi}

As an example, the linking permutation of $5\,2\,4\,1\,3$ is
$4\,5\,1\,2\,3$.
Observe that $\varphi \in \mathbf{S}_n^c$ follows from
$\varphi(\pi^{-1}(i))=\pi^{-1}(i+1)$, where $n+1 \equiv 1$.
The linking permutation of a BW-array gives the ranks of the
consecutive shifts in the lexicographic order.
Furthermore,
note that if $\pi(1)$ and $\varphi = \Phi(\pi)$ are known, then one can
reconstruct $\pi$ by iterating $\pi(\varphi(i)) = \pi(i)+1$ starting
with $i=1$.

\begin{defi}
For two permutations $\pi$ and $\sigma$ of $\mathbf{S}_n$,
define $\pi \sim \sigma$ if and
only if there exists $k \in [1,n]$ such that $\sigma = \pi+k$.
\end{defi}

It is easily seen that $\sim$ is an equivalence relation on $\mathbf{S}_n$.

\begin{prop} \label{phi_jell}
$\Phi$ is well-defined on ${\mathbf{S}_n}/\!\sim$ and bijective from
${\mathbf{S}_n}/\!\sim$ to $\mathbf{S}_n^c$.
\end{prop}
\begin{proof}
We first show that $\Phi(\pi)=\Phi(\sigma)$ if $\pi\sim\sigma$.
Let $\sigma=\pi+k$.
Then, $\Phi(\pi) = \pi^{-1}(\pi+1) = \pi^{-1}(((\pi+k)+1)-k)=
\pi^{-1}(\ID-k)((\pi+k)+1)$.
Observe now that
$(\pi+k)^{-1}=\pi^{-1}(\ID-k)$. Therefore,
$\Phi(\pi)=(\pi+k)^{-1}((\pi+k)+1)=\Phi(\pi+k)=\Phi(\sigma)$.
This shows that $\Phi$ is well-defined on ${\mathbf{S}_n}/\!\sim$.

Second, we can uniquely determine $\pi$ from $\Phi(\pi)$ and $\pi(1)$,
and hence $\Phi$ is injective on ${\mathbf{S}_n}/\!\sim$.
The number of $\sim$-equivalence classes is $(n-1)! =
|\mathbf{S}_n^c|$.
Therefore, the mapping is a bijection.
\end{proof}

\begin{defi}[permutation descent]
Let $\pi\in\mathbf{S}_n$. We say that $i \in [1,n-1]$ is a
\emph{descent} of $\pi$ if and only if $\pi(i)>\pi(i+1)$.
The set of all descents of $\pi$ is denoted
$\Des(\pi)$.
\end{defi}

The following theorem from~\cite{Cro05} provides a nice
characterization of BW-arrays through the linking mapping.
It will play a central role in our study.

\begin{theorem}[\cite{Cro05}] \label{BW_jell_erosebb}
Let $r_i \geq 0$, $1\leq i\leq k$, be integers such that
$\sum^k_{i=1}r_i=n$.
A permutation $\pi \in \mathbf{S}_n$ is the
BW-array of a primitive word $w \in \Sigma^n$
with $r_i$ occurrences of letter $a_i$, $1\leq i\leq k$,
if and only if
$\Des(\Phi(\pi))\subseteq\{r_1\,,r_1+r_2\,,\dotsc,\, r_1+\dots
+r_{k-1} \}$.
Moreover, in this case $\pi$ is the BW-array of exactly one such word.
\end{theorem}

\section{Characterization of suffix arrays}
\label{sect-charact}
In this section, we state our characterization theorems for suffix arrays:
Theorems~\ref{suff_jell_erosebb_vegejel_nelkul},
\ref{theorem_without_endline_symbol} and
~\ref{th:bijekcio_fivel}. We use a reduction of suffix sorting to
cyclic shift sorting by appending a sentinel symbol to the end of the word, and thereby reduce the characterization of suffix arrays to the characterization of BW-arrays.


Consider a symbol $\sharp\notin \Sigma$, and the alphabet
$\Sigma'=\{\sharp,a_1,a_2,\dots,a_k\}$ with $\sharp<a_1<a_2\dots
<a_k$. We will examine the suffix arrays of words $w\sharp$ for
$w\in\Sigma^n$. The following proposition is obvious.

\begin{prop} \label{claim:vegeszimb_ekvivalencia}
There is a one-to-one correspondence between the suffix arrays of
$w\in\Sigma^n$ and the suffix arrays of $w'\in\Sigma^n\sharp$. If
$\sigma\in \mathbf{S}_n$ is the suffix array of $w$, then $\pi\in
\mathbf{S}_{n+1}$ is the suffix array of $w\sharp$ if and only if
$\pi=(n+1)\,\sigma(1)\,\sigma(2)\dots\sigma(n)$.
\end{prop}

The following proposition shows, that for words in $\Sigma^n\sharp$, cyclic shift
sorting is equivalent to suffix sorting.
Note that this property remains true even if we do not assume
that $\sharp$ is the smallest element in the ordering of $\Sigma\cup \{\sharp\}$.
This will be important later.

\begin{prop} \label{mikor_egyenlo_BW_suffix}
\label{prop3}
Let $w'=w\sharp$, where $w\in\Sigma^*$, $\sharp\notin\Sigma$ and the ordering of $\Sigma\cup\{\sharp\}$ is arbitrary.
Then $w'$ is primitive, hence the order of its cyclic shifts is
well-defined.
Moreover the order of the cyclic shifts of $w\sharp$
coincides with the order of the suffixes of $w\sharp$.
\end{prop}
\begin{proof}
If $w'=u^k$ for some word $u$ and a $k>1$, then $k$ would divide the
number of occurrences of $\sharp$. Therefore $w'$ is primitive.

As $w'$ has only one occurrence of $\sharp$, then in comparing two different
cyclic shifts we necessarily compare $\sharp$ with some other
character. This means that the lexicographic order of two cyclic
shifts is decided no later than at the position of the first
$\sharp$. Therefore if we leave out the characters after the
$\sharp$ in both shifts, we get the same ordering.
\end{proof}



Now, we give two theorems characterizing the permutations that are
suffix arrays for some word $w\sharp$ with $w\in\Sigma^n$.

\begin{theorem} \label{suff_jell_erosebb}
Let $r_i\geq 0$, $1\leq i\leq k$, be integers such that
$\sum^k_{i=1}r_i=n$. A permutation $\pi \in \mathbf{S}_{n+1}$ is the suffix array of a word
$w\sharp$ with $w\in \Sigma^n$ with $r_i$
occurrences of the letter $a_i$, $1\leq i\leq k$,
if and only if
$\Des(\Phi(\pi))\subseteq\{1,1+r_1,1+r_1+r_2,\dots,
1+r_1+\dots +r_{k-1} \}$ and $\pi(1)=n+1$.
Moreover, in this case, $\pi$ is the suffix array of exactly one such
word.
\end{theorem}
\begin{proof}
According to Theorem~\ref{BW_jell_erosebb}, $\pi \in \mathbf{S}_{n+1}$
is the BW-array of a primitive word $w\in (\Sigma\cup\sharp)^{n+1}$
which has $r_i$ occurrences of letter $a_i$
and one occurrence
of symbol $\sharp$, if and only if the first condition is satisfied, and in
this case there is only one such primitive word. Here the primitivity
is immediate, since we have only one occurrence of $\sharp$.
%
Since $\sharp$ is the smallest letter, condition $\pi(1)=n+1$ is
necessary and sufficient for $\sharp$ to be
the last letter. Finally the BW-array coincides with the suffix array
on the class of words type $w\sharp$, by Proposition~\ref{prop3}.
\end{proof}

From Theorem~\ref{suff_jell_erosebb}, we can easily deduce the
following theorem: 

\begin{theorem}
\label{theor3}
A permutation $\pi\in\mathbf{S}_{n+1}$ is the suffix array of a word
$w\sharp$ where $w\in\Sigma^n$ if and only if
(i) $|\Des(\Phi(\pi))\setminus\{1\}|\leq k-1$, and
(ii) $\pi(1)=n+1$.
\end{theorem}
\begin{proof}
Let $\pi\in\mathbf{S}_{n+1}$ be the suffix array of a word $w\sharp$
for $w\in\Sigma^n$. Assume $w$ has $r_i\geq 0$ occurrences of
letter $a_i$ for each $i\in[1,k]$. Then
conditions $(i)$ and $(ii)$ follow immediately from
Theorem~\ref{suff_jell_erosebb}.
Conversely, let $\Des(\Phi(\pi))\setminus\{1\}=\{d_1,d_2,\dotsc,d_\ell\}$ for
$\ell\leq k-1$. Then for $r_1=d_1-1$, $r_2=d_2-d_1$, \ldots,
$r_\ell=d_\ell-d_{\ell-1}$, $r_{\ell+1}=\dotsc=r_{k-1}=0$, we have
$\Des(\Phi(\pi))\subseteq\{1,1+r_1,1+r_1+r_2,\dots, 1+r_1+\dots
+r_{k-1} \}$. $\pi(1)=n+1$ is also satisfied.
Then, by Theorem~\ref{suff_jell_erosebb}, there is a word $w\sharp$ with the
corresponding numbers of letter occurrences that has $\pi$ as its
suffix array.
\end{proof}

Now we provide a characterization of suffix arrays for the case where we do not
assume a sentinel symbol at the end of the word.
Proposition~\ref{claim:vegeszimb_ekvivalencia} combined with
Theorem~\ref{suff_jell_erosebb} and Theorem~\ref{theor3} respectively imply the following results.

\begin{theorem} \label{suff_jell_erosebb_vegejel_nelkul}
Let $r_i\geq 0$, $1\leq i\leq k$, be integers such that $\sum^k_{i=1}r_i=n$.
A permutation $\pi \in \mathbf{S}_n$ is the suffix array of a word
$w\in \Sigma^n$ with $r_i$ occurrences of the letter $a_i$,
$1 \leq i \leq k$,
if and only if, for $\pi'=(n+1)\, \pi(1)\dots \pi(n)$,
$\Des(\Phi(\pi'))\subseteq\{1,1+r_1,1+r_1+r_2,\dots, 1+r_1+\dots
+r_{k-1} \}$. Moreover, in this case $\pi$ is the suffix array of
exactly one such word.
\end{theorem}

\begin{theorem} \label{theorem_without_endline_symbol}
A permutation $\pi\in\mathbf{S}_n$ is the suffix array of some word
$w\in\Sigma^n$ if and only if, for $\pi'=(n+1)\, \pi(1)\dots \pi(n)$, we
have $|\Des(\Phi(\pi'))\setminus\{1\}|\leq k-1$.
\end{theorem}

Finally we give a result stating a
bijection between the suffix arrays over an alphabet $\Sigma$ and a certain set of permutations.

\begin{theorem} \label{th:bijekcio_fivel}
For a permutation $\pi\in\mathbf{S}_n$, let $\pi'=(n+1)\,\pi(1)\dots\pi(n)$.
The mapping $\pi \mapsto \Phi(\pi')$ is a bijection between the suffix
arrays of words $w\in\Sigma^n$
and the permutations $\varphi\in\mathbf{S}_{n+1}^c$ with
$|\Des(\varphi)\setminus\{1\}|\leq k-1$.
Moreover,
given such a $\varphi\in\mathbf{S}_{n+1}^c$, we can easily compute the
corresponding suffix array $\pi$ as follows:
$\pi^{-1}(i)=\varphi^{i}(1)-1$ for each $i\in[1,n]$.
\end{theorem}
\begin{proof}
Let us denote the $\sim$ equivalence class of a
$\sigma\in\mathbf{S}_{n+1}$ by $[\sigma]$.
The mapping
$f:\mathbf{S}_n\rightarrow\mathbf{S}_{n+1}/\sim$
defined by $f(\pi)=[\pi']$ is a bijection between $\mathbf{S}_n$ and
$\mathbf{S}_{n+1}/\sim$.
$\Phi$ is bijective from ${\mathbf{S}_{n+1}}/\!\sim$ to $\mathbf{S}_{n+1}^c$,
and hence $\pi\mapsto\Phi(\pi')$ is a bijection from $\mathbf{S}_n$ to
$\mathbf{S}_{n+1}^c$.
According to Theorem~\ref{theorem_without_endline_symbol},
a permutation $\pi\in\mathbf{S}_n$ is a suffix array of some word
in $\Sigma^n$
if and only if, for its its image $\Phi(\pi')$,
$|\Des(\Phi(\pi'))\setminus\{1\}|\leq k-1$.
Then it follows that the restriction
of the mapping to the set of suffix permutations gives a bijection
into the set of permutations $\varphi\in\mathbf{S}_{n+1}^c$ with
$|\Des(\varphi)\setminus\{1\}|\leq k-1$.

As for the computation of the inverse mapping, we know that
$(\pi')^{-1}(n+1)=1$ and that $\Phi(\pi')((\pi')^{-1}(i))=(\pi')^{-1}(i+1)$.
Therefore, if $\varphi=\Phi(\pi')$, then
$\pi^{-1}(i)=(\pi')^{-1}(i)-1=\varphi^{i}(1)-1$ for all
$i\in[1,n]$.
\end{proof}

\section{Relation to the characterization of He et
  al}

He et al. \cite{He2005} proposed a characterization of suffix arrays
for a binary alphabet $\Sigma=\{a,b\}$ in the special case where the
sentinel character $\sharp$ is ranked between the characters of
$\Sigma$, \emph{i.e.}, $a < \sharp < b$.
In this case, the
lexicographic order of suffixes of $w$ can be different from the
lexicographic order of the corresponding suffixes of $w\sharp$,
therefore this definition gives a slightly different suffix array notion.

In this section, we elucidate how the characterization of \cite{He2005} is
related to our characterizations given in Section~\ref{sect-charact}.
In particular, we show that our approach yields a simpler
characterization that implies the result of \cite{He2005}.
Before describing the characterization of~\cite{He2005}, we show
that Theorem~\ref{BW_jell_erosebb} allows us to obtain a
characterization of suffix arrays for this kind of alphabet ordering
as well, similarly to the usual ordering of the previous section.

\begin{theorem}
\label{our-theorem}
A permutation $\pi\in \mathbf{S}_{n+1}$ is the suffix array of a word
$w\sharp$ with $w\in\{a,b\}^n$ and $a < \sharp < b$ if and only if
$\Des(\Phi(\pi))\subseteq\{\pi^{-1}(n+1)-1,\pi^{-1}(n+1)\}$.
\end{theorem}
\begin{proof}
Let $\pi \in \mathbf{S}_{n+1}$. By Proposition
\ref{mikor_egyenlo_BW_suffix},
$\pi$ is the suffix array of $w\sharp$ if and only
if it is the BW-arrays of $w\sharp$. Therefore it is enough to prove the
theorem for BW arrays instead of suffix arrays.

We first show the 'only if' part. Observe that if $\pi$ is the BW array of $w\sharp$, then $w_{\pi(i)}=a$ for
$i<\pi^{-1}(n+1)$, and $w_{\pi(i)}=b$ for $i >\pi^{-1}(n+1)$.
Therefore $w$ has $\pi^{-1}(n+1)-1$ occurrences of $a$, 1 occurrence of $\sharp$ and
$n+1-\pi^{-1}(n+1)$ occurrences of $b$.
By Theorem \ref{BW_jell_erosebb}, we immediately obtain
$\Des(\Phi(\pi))\subseteq\{\pi^{-1}(n+1)-1,\pi^{-1}(n+1)\}$.

We now prove the 'if' part. Suppose that $\Des(\Phi(\pi))\subseteq\{\pi^{-1}(n+1)-1,\pi^{-1}(n+1)\}$.
From Theorem \ref{BW_jell_erosebb} there exists exactly one word $w'\in \{a,\sharp,b\}^{n+1}$ that has $\pi^{-1}(n+1)-1$ occurrences of $a$, 1 occurrence of $\sharp$ and $n+1-\pi^{-1}(n+1)$ occurrences of $b$ and which has $\pi$ as BW array. From $w'_{\pi(1)}\leq w'_{\pi(2)}\leq \dots\leq w'_{\pi(n+1)}$, this word is the following:
$w'_{\pi(i)}=a$ for $i<\pi^{-1}(n+1)$,
$w'_{\pi(\pi^{-1}(n+1))}=w'_{n+1}=\sharp$, and
$w'_{\pi(i)}=b$ for $i >\pi^{-1}(n+1)$.
We can see, that $w'=w\sharp$ where $w\in\{a,b\}^n$, therefore we have the sufficiency of the condition.

\end{proof}

Now, we repeat the characterization given by He et al. \cite{He2005}. We need some additional definitions.

\begin{defi}
[Ascending-to-max \cite{He2005}]
A permutation $\pi\in \mathbf{S}_{n+1}$ is \emph{ascending-to-max} if
and only if, for every $i\in[1,n-1]$, we have
\begin{enumerate}
\item[(a)] if $\pi^{-1}(i)<\pi^{-1}(n+1),\ \pi^{-1}(i+1)<\pi^{-1}(n+1)$, then $\pi^{-1}(i) < \pi^{-1}(i+1)$, and
\item[(b)] if $\pi^{-1}(i)>\pi^{-1}(n+1),\ \pi^{-1}(i+1)>\pi^{-1}(n+1)$, then $\pi^{-1}(i) > \pi^{-1}(i+1)$.
\end{enumerate}
\end{defi}

\begin{defi}
[Non-nesting \cite{He2005}]
A permutation $\pi\in \mathbf{S}_{n+1}$ is \emph{non-nesting} if and
only if, for each $i,j\in[1,n]$ such that $\pi^{-1}(i) < \pi^{-1}(j)$, if
\begin{enumerate}
\item[(a)]  $\pi^{-1}(i)<\pi^{-1}(i+1) \quad $ and $\quad \pi^{-1}(j)<\pi^{-1}(j+1)$, or
\item[(b)]  $\pi^{-1}(i)>\pi^{-1}(i+1)\quad $ and $\quad \pi^{-1}(j)>\pi^{-1}(j+1)$,
\end{enumerate}
then $\pi^{-1}(i+1) < \pi^{-1}(j+1)$.
\end{defi}

The characterization of \cite{He2005} is as follows.

\begin{theorem}[\cite{He2005}]
\label{He-theorem}
A permutation $\pi\in \mathbf{S}_{n+1}$ is the suffix array of a word
$w\sharp$ with $w\in\{a,b\}^n$ and $a < \sharp < b$ if and only if it
is both ascending-to-max and non-nesting.
\end{theorem}

We now show that the condition of Theorem~\ref{He-theorem} is equivalent
to that of Theorem~\ref{our-theorem}.
Let $\varphi = \Phi(\pi)$.
We have $\varphi(\pi^{-1}(i))\!=\pi^{-1}(i+1).$
Therefore, the ascending-to-max property reduces to
$i < \varphi(i)$ for $i \in[1,\pi^{-1}(n+1)-1]$, and
$i > \varphi(i)$
for $i\in [\pi^{-1}(n+1)+1,n+1]$.
As for the non-nesting property, we have the following:
for $i,j\in[1,n]\setminus\{\pi^{-1}(n+1)\}$,
if $i < \varphi(i)$ and $j < \varphi(j)$, or
$i > \varphi(i)$ and $j > \varphi(j)$,
then $i<j$ implies $\varphi(i)<\varphi(j)$.

We show that the two conditions together are equivalent to the condition of
Theorem~\ref{our-theorem}.
The two conditions together trivially imply the condition of
Theorem~\ref{our-theorem}.
Conversely,
suppose that $\Des(\Phi(\pi))\subseteq\{\pi^{-1}(n+1)-1,\pi^{-1}(n+1)\}$.
$\varphi$ has one orbit, and hence $\varphi(1)>1$. If for some $j\in[1,n]$ we have
$\varphi(j)>j$ and $ \varphi(j+1)<j+1$, then $j$ is a descent.
Hence for $i\in[1,\pi^{-1}(n+1)-1]$, $i < \varphi(i)$. Similarly,for
$i \in [\pi^{-1}(n+1)+1,n+1]$, $i > \varphi(i)$.
From $\Des(\Phi(\pi))\subseteq\{\pi^{-1}(n+1)-1,\pi^{-1}(n+1)\}$ it follows
that $\varphi$ is monotone on $[1,\pi^{-1}(n+1)-1]$ and on $[\pi^{-1}(n+1)+1,n+1]$,
and hence the non-nesting property is also satisfied.


\section{Enumerations} \label{szamolasok}

Our characterization theorems from Section~\ref{sect-charact} can also
be used to count objects related to suffix arrays. Sch\"urmann and Stoye
\cite{Schurmann2008} obtained some counting results using ``direct''
combinatorial considerations. Here we give shorter proofs of these
results, based on bijections provided by
Theorems~\ref{suff_jell_erosebb_vegejel_nelkul}-\ref{th:bijekcio_fivel}
from Section~\ref{sect-charact}. The underlying ideas of the proofs are the same as in \cite{Schurmann2008}, the simplification is due to the more simple characterization of suffix arrays.

The following enumerations have been studied in \cite{Schurmann2008}.
\begin{enumerate}
 \item[(i)] For a permutation $\pi\in \mathbf{S}_n$, count the number
   of words of length $n$ over an alphabet of size $k$ that have $\pi$
   as their suffix array,
 \item[(ii)] For a permutation $\pi\in \mathbf{S}_n$, count the number
   of words of length $n$ over an alphabet of size $k$, that have at
   least one
   occurrence of each letter and have $\pi$ as their suffix array,
 \item[(iii)] Count the number of permutations $\pi\in \mathbf{S}_n$
   that are suffix arrays of some word over an alphabet of size $k$.
\end{enumerate}


We start with question (i).

\begin{theorem}[\cite{Schurmann2008}]
\label{enumeration1}
 For a permutation $\pi\in\mathbf{S}_n$, let $\pi'=(n+1)\,\pi(1)\dots\pi(n)$.
 The number of words of length $n$ over an alphabet of size $k$ having
 $\pi$ as their suffix array is
 $$\binom{n+k-1-|\Des(\Phi(\pi'))\setminus\{1\}|}{k-1-|\Des(\Phi(\pi'))\setminus\{1\}|}.$$
\end{theorem}
\begin{proof}
Theorem~\ref{suff_jell_erosebb_vegejel_nelkul} states that if $r_i\geq
0$ for $i=1\dots k$, $\sum_{i=1}^{k} r_i=n$, and
\begin{equation}\label{eq:felt}\Des(\Phi(\pi'))\subseteq
  \{1,1+r_1,1+r_1+r_2,\dots, 1+r_1+\dots +r_{k-1} \},
\end{equation}
then there is exactly one word $w$ with $r_i$ occurrences of $a_i$
that has $\pi$ as its suffix array. Therefore, we
need to count
the number of tuples $(r_1, \dots,r_k)$ (Parikh vectors) that satisfy
(\ref{eq:felt}) given a permutation $\pi'=(n+1)\,\pi(1)\dots\pi(n)$.
We represent a tuple $(r_1, \dots,r_k)$, $\sum^k_{i=1} r_i=n$, as a
sequence of $n$ dots divided into $k$ (possibly empty) groups separated by $k-1$ separators:
\begin{equation}\label{dots}
(r_1,\dots, r_k),  \sum r_i=n \quad \leftrightarrow \quad
\underbrace{\circ \dots\circ}_{r_1} |\underbrace{\circ \dots
  \circ}_{r_2} | \dots | \underbrace{\circ  \dots \circ}_{r_k}
\end{equation}
Clearly, this representation is a bijection.

If $i>1$ is a descent of $\Phi(\pi')$, there must be a separator
between the $(i-1)$-th and the $i$-th dots. This defines the placement
of $|\Des(\Phi(\pi'))\setminus\{1\}|$ separators. The remaining
$(k-1-|\Des(\Phi(\pi'))\setminus\{1\}|)$ separators can interleave the $n$
dots arbitrarily. This can be done in
$$\binom{n+k-1-|\Des(\Phi(\pi'))\setminus\{1\}|}{k-1-|\Des(\Phi(\pi'))\setminus\{1\}|}=\binom{n+k-1-|\Des(\Phi(\pi'))\setminus\{1\}|}{n}$$
ways. It is easy to get convinced that the construction provides a
bijection with the considered set of tuples. The result follows.
\end{proof}

Note that if $k-1<|\Des(\Phi(\pi'))\setminus\{1\}|$, there is no word on an
alphabet of size $k$ which has $\pi$ as its suffix array. This is
confirmed by Proposition~\ref{enumeration1}, as $\binom{m}{n}=0$ for
$m<n$.
The following proposition from \cite{Schurmann2008} answers question (ii).

\begin{theorem}[\cite{Schurmann2008}]
For a permutation $\pi\in\mathbf{S}_n$, let
$\pi'=(n+1)\,\pi(1)\dots\pi(n)$.
The number of words of length $n$ over an alphabet of size $k$ that
have at least one occurrence of each of the $k$ letters and have
$\pi$ as their suffix array
is $$\binom{n-1-|\Des(\Phi(\pi'))\setminus\{1\}|}{k-1-|\Des(\Phi(\pi'))\setminus\{1\}|}.$$
\end{theorem}
\begin{proof}
We modify the proof of Proposition~\ref{enumeration1} to insure that
that each letter occurs at least once. We cannot have two adjacent
separators, and we cannot start or end with a separator.
We then have to distribute $k-1$ separators among $n-1$ possible
places between the circles. 
Like in the proof of Proposition~\ref{enumeration1}, the place of
$|\Des(\Phi(\pi'))\setminus\{1\}|$ separators is determined by $\pi'$, and
the remaining $(k-1-|\Des(\Phi(\pi'))\setminus\{1\}|)$ separators are
distributed among $(n-1-|\Des(\Phi(\pi'))\setminus\{1\}|)$ remaining
places. This yields the count of the Theorem.
%
%
\end{proof}

Finally, we give a proof for question (iii), based on the results of
Section~\ref{sect-charact}. Let $\Big\langle \begin{matrix} n
  \\ d \end{matrix} \Big\rangle$ denote the Eulerian number,
  i.e. the number of permutations of $[1,n]$ with exactly $d$
  descents.

\begin{theorem}[\cite{Schurmann2008}]
The number of permutations $\pi\in \mathbf{S}_n$ that are suffix arrays of a word $w \in \Sigma^n$ with $|\Sigma|=k$ is ${\displaystyle  \sum_{d=0}^{k-1}\Big\langle \begin{matrix} n \\ d \end{matrix} \Big\rangle}$.
\end{theorem}
\begin{proof}
According to Theorem~\ref{th:bijekcio_fivel}, there is a bijection
between the suffix arrays of words $w\in\Sigma^n$
and the permutations $\varphi\in\mathbf{S}_{n+1}^c$ such that
$|\Des(\varphi)\setminus\{1\}|\leq k-1$. We then have to count the
number of such permutations.
Let $P(n,d)$ denote the number of permutations $\varphi\in
\mathbf{S}_{n+1}^c$ with $|\Des(\varphi)\setminus\{1\}|= d$.
To prove the theorem, we show that
$P(n,d)$ is equal to the Eulerian number
$\Big\langle \begin{matrix} n \\ d \end{matrix} \Big\rangle$.

%
The proof is by induction on $n$.
Trivially, $P(1,0)\!=\!1\!=\!\big\langle \begin{smallmatrix}\! 1\!
  \\ \! 0 \!\end{smallmatrix} \big\rangle\ $ (the only good
permutation is $\pi=2\,1$), and
$P(1,d)\!=\!0\!=\!\big\langle \begin{smallmatrix}\! 1\! \\ \!
  d\! \end{smallmatrix} \big\rangle\ $ when $d\geq 1$.
We now show that $P(n,d)=(d+1)P(n-1,d)+(n-d)P(n-1,d-1)$,
thereby proving that
$P(n,d)=\Big\langle \begin{matrix} n \\ d \end{matrix}
\Big\rangle$.
For the inductive step, we describe a generative procedure for the
considered permutations.
Consider $\varphi\in\mathbf{S}_{n}^c$ and
let $s\in[2,n+1]$. Consider the mapping
$\Aug_s:[1,n]\rightarrow [1,n+1]$ defined by
$$
\Aug_s(i)=
\begin{cases}
i &\textrm{ if } i<s, \\
i+1 & \textrm{ if } i\geq s \text{.}
\end{cases}
$$
Observe that $\Aug_s\circ\varphi\circ\Aug_s^{-1}$ is bijective on the
set $\{1,\dotsc ,s-1, s+1, \dotsc ,n\}$ and has one orbit.
Now consider the transform
$T_s:\mathbf{S}_{n}^c\rightarrow\mathbf{S}_{n+1}^c$ defined by
$$
T_s(\varphi)(i)=
\begin{cases}
\Aug_s\circ\varphi\circ\Aug_s^{-1}(i) &\textrm{ if }i\in[1,n+1]\setminus\{1,s\}, \\
s &\textrm{ if }i=1, \\
\Aug_s\circ\varphi\circ\Aug_s^{-1}(1)=\varphi(1) & \textrm{ if } i=s \text{.}
\end{cases}
$$
It is straightforward to check that $T_s(\varphi)\in\mathbf{S}^c_{n+1}$,
\emph{i.e.}, $T_s(\varphi)\in\mathbf{S}_{n+1}$ and it has one orbit.
The
construction is illustrated in Figure~\ref{figure1}.

\begin{figure}[t]
    \begin{tikzpicture}[scale=1.,label distance=5mm,>=latex]%
      \tikzstyle{vertex}=[draw,circle,thick,inner sep=.75mm],
      \tikzstyle{arc}=[->,thick,rounded corners],
      \begin{scope}
        \foreach \x in {1,2,3,4}
        \filldraw (\x,0) node (v\x) [vertex] {$\x$};
        \draw (2.5,-1.5) node {$\varphi$};
        \draw[arc]
        (v1) -- +(0.1,0.7) -- +(1.9,0.7) -- (v3);
        \draw[arc]
        (v3) -- +(0.1,0.7) -- +(0.9,0.7) -- (v4);
        \draw[arc]
        (v4) -- +(-0.1,-0.7) -- +(-1.9,-0.7) -- (v2);
        \draw[arc]
        (v2) -- +(-0.1,-0.7) -- +(-1.0,-0.7) -- (v1);
      \end{scope}
      \begin{scope}[xshift=4cm]
        \foreach \x in {1,2,4,5}
        \filldraw (\x,0) node (v\x) [vertex] {$\x$};
        \draw (3,-1.5) node {$\Aug_s\circ \,\varphi \circ \Aug_s^{-1}$};
        \draw[arc]
        (v1) -- +(0.1,0.7) -- +(2.9,0.7) -- (v4);
        \draw[arc]
        (v4) -- +(0.1,0.7) -- +(0.9,0.7) -- (v5);
        \draw[arc]
        (v5) -- +(-0.1,-0.7) -- +(-2.9,-0.7) -- (v2);
        \draw[arc]
        (v2) -- +(-0.1,-0.7) -- +(-1.0,-0.7) -- (v1);
      \end{scope}
      \begin{scope}[xshift=9cm]
        \foreach \x in {1,2,3,4,5}
        \filldraw (\x,0) node (v\x) [vertex] {$\x$};
        \draw (3,-1.5) node {$T_s(\varphi)$};
        \draw[arc,ultra thick]
        (v1) -- +(0.1,0.7) -- +(1.9,0.7) -- (v3);
        \draw[arc,ultra thick]
        (v3) -- +(0.1,0.7) -- +(0.9,0.7) -- (v4);
        \draw[arc]
        (v4) -- +(0.1,0.7) -- +(0.9,0.7) -- (v5);
        \draw[arc]
        (v5) -- +(-0.1,-0.7) -- +(-2.9,-0.7) -- (v2);
        \draw[arc]
        (v2) -- +(-0.1,-0.7) -- +(-1.0,-0.7) -- (v1);
      \end{scope}
    \end{tikzpicture}
    \caption{\label{figure1}%
      Illustration of $T_s(\varphi)$ with $\varphi=3142$ and
      $s=3$. Informally, $\Aug_s\circ\varphi\circ\Aug_s^{-1}$ ``increments by
      one'' all nodes $s,\ldots,n$. Then $T_s(\varphi)$ ``splits'' the
      mapping $(1,\varphi(1))$ into $(1,s)$ and $(s,\varphi(1))$.}
\end{figure}
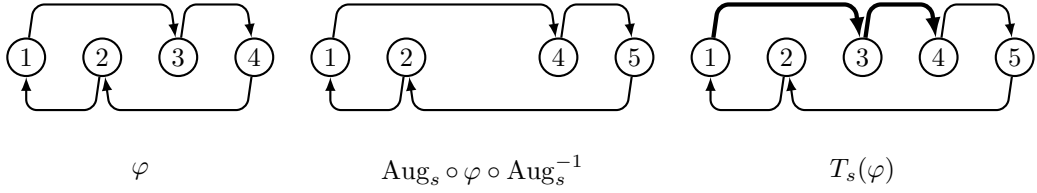

Furthermore, if $r,s\in[2,n+1]$ and $\varphi,\psi\in \mathbf{S}_n^c$,
where $r\neq s$ or $\varphi \neq \psi$, then $T_s(\varphi)\neq
T_r(\psi)$.
%
As there are $(n-1)!$ permutations in $\mathbf{S}_n^c$, we get
$n\cdot(n-1)!=n!$ different permutations $T_s(\varphi)$ for
$s\in[2,n+1]$ and $\varphi\in \mathbf{S}_n^c$. Therefore,
$\mathbf{S}_{n+1}^c=\{T_s(\varphi)\ :\ s\in[2,n+1],\ \varphi\in
\mathbf{S}_n^c\}$.

Now, we examine how the transform $T_s(\varphi)$ affects the number of
descents of $\varphi$.
For $i\in [2,n-1]\setminus\{s-1\}$,
$T_s(\varphi)(\Aug_s(i))=\Aug_s(\varphi(i))$ and
$T_s(\varphi)(\Aug_s(i)+1)=T_s(\varphi)(\Aug_s(i+1))=\Aug_s(\varphi(i+1))$. Therefore
for $i\in [2,n-1]\setminus\{s-1\}$
$$
\Aug_s(i)\in\Des(T_s(\varphi))
\;\Leftrightarrow\;
i\in\Des(Aug_s(\varphi))
\;\Leftrightarrow\;
i\in\Des(\varphi),
$$
where the second equivalence follows from the monotonicity of $\Aug_s$.
Thus, $\Aug_s$ gives a one-to-one correspondence between
$\Des(\varphi)\setminus\{1,s-1\}$ and
$\Des(T_s(\varphi))\setminus\{1,s-1,s\})$. It remains to analyze
values $s-1$ and $s$.
We have
$T_s(\varphi)(s+1)=Aug_s(\varphi)(s)<Aug_s(\varphi)(s-1)=T_s(\varphi)(s-1)$
if and only if $s-1\in \Des(\varphi)$. In this case, $s-1$ or $s$ is a
descent of $T_s(\varphi)$.
The insertion of $T_s(\varphi)(s)=\varphi(1)$ may or may not create a new descent.
For a given $\varphi\in\mathbf{S}_n^c$, in each monotonic run of
$\varphi$ on indices $\{2,\dots, n\}$, there is exactly one position
where we can place $\varphi(1)$ without creating a new descent,
otherwise we create exactly one new descent.

How many $T_s(\varphi)$ can we have with $|\Des(T_s(\varphi))\setminus\{1\}|=d$?
For each $\varphi\in\mathbf{S}_n^c$ with
$|\Des(\varphi)\setminus\{1\}|=d$, we have $(d+1)$ possibilities to
choose $s$ ($\varphi$ has $d+1$ monotonic runs on $\{2,\dots, n\}$).
For each $\varphi\in\mathbf{S}_n^c$ with
$|\Des(\varphi)\setminus\{1\}|=d-1$, we have $(n-d)$ possibilities to
choose $s$.
These permutations are all different as $T_s(\varphi)\neq T_r(\psi)$
if $s\neq r$ or $\varphi \neq \psi$. There is no other way to get a
permutation $\psi\in\mathbf{S}_{n+1}^c$ with
$|\Des(\psi)\setminus\{1\}|=d$. We conclude that
$P(n,d)=(d+1)P(n-1,d)+(n-d)P(n-1,d-1)$. This proves the Theorem.
\end{proof}


\bibliographystyle{elsarticle-num-names}

\end{document}